\documentclass[journal]{IEEEtran}
\ifCLASSINFOpdf
\else
\fi
\usepackage{amssymb}
\usepackage{amsmath}
\usepackage{algorithm}
\usepackage{algorithmic}
\usepackage{hyperref}
\usepackage{tikz}
 
\usepackage{graphicx}
\usepackage{booktabs}
\usepackage{xspace}

\usepackage{multirow}
\usepackage{makecell}

\usepackage{ mathrsfs }
\usepackage{graphicx, adjustbox}
\usepackage{array}
\usepackage{setspace}
\usepackage{enumitem}
\usepackage{breqn}
\usepackage{bm}
\usepackage{ dsfont }
\usepackage{ upgreek }
\usepackage{textcomp}
\usepackage{multirow}
\usepackage{algorithm}
\usepackage{algorithmic}
\usepackage{caption}
\usepackage{pgfplots}
\usepackage{tkz-euclide}
\usepackage{subcaption}
\usepackage{tikz}{\tiny }
\usepackage{balance}
\usepackage{ tipa }

\usepackage{fancyhdr}
\usepackage{ upgreek }
\usepackage{soul}
\usepackage{breqn}
\usepackage[english]{babel}
\usepackage{braket} 

\addto\captionsenglish{}
\allowdisplaybreaks

\usepackage{amsthm}
\newtheorem{theorem}{Theorem}
\newtheorem{lemma}{Lemma}

\newtheorem{corollary}{Corollary}

\newtheorem{remark}{Remark}

\newtheorem{assumption}{Assumption}

\usepackage{fancyhdr}
\usepackage{ upgreek }
\usepackage{soul}
\makeatletter
\newcommand{\vast}{\bBigg@{4}}

\newcommand{\Vast}{\bBigg@{5}}
\begin{document}

\title{Geometric Fairness-Aware Routing for Federated Edge Networks}

\author{Ratun Rahman

\thanks{Ratun Rahman is with the Department of Electrical and Computer Engineering, University of Alabama in Huntsville, Huntsville, AL 35899, USA, email: rr0110@uah.edu} 
}



\markboth{accepted in the IEEE/ACM Transactions on Networking}%
{Shell \MakeLowercase{\textit{et al.}}: Bare Demo of IEEEtran.cls for IEEE Journals}
\maketitle

\begin{abstract}
Emerging 6G and edge-intelligent networks require effective and balanced routing algorithms among varied and spatially distributed devices.   Existing federated routing systems often prioritize aggregate latency or throughput above fairness and the underlying geometric structure of network topologies.   This paper describes Geo-FairFed, a geometric fairness-aware routing system that blends hyperbolic graph neural networks (HGNNs) and federated optimization to provide equal performance across edge nodes.   Each node learns topology-aware representations on a negatively curved manifold, which include hierarchical relationships and connection asymmetries.
A global aggregator next enforces fairness using a curvature-regularized aim that minimizes routing loss, geometric inconsistency, and an inequality penalty based on Jain's fairness index.   A theoretical analysis develops convergence guarantees under limited curvature and shows that the proposed fairness term results in a Pareto-improving equilibrium in routing performance.   Extensive simulations on dynamic 6G-edge and IoT topologies reveal that Geo-FairFed minimizes average latency by 20\%, reduces energy consumption by 17\%, and improves fairness by up to 21\% when compared to state-of-the-art federated and geometric routing protocols.   The study found that embedding topology in a hyperbolic manifold and including fairness into federated updates can significantly enhance the efficiency and equity of large-scale network routing.
\end{abstract}

\maketitle



\section{Introduction}\label{sec:introduction}

\IEEEPARstart{T}{he}
increasing popularity of 6G and edge-intelligent networks is reshaping routing and resource management in large-scale communication systems.   In contrast to standard centralized architectures, future networks will consist of billions of diverse, geographically spread edge devices that must make swift, autonomous choices in unpredictable and dynamic situations \cite{bogyrbayeva2024machine}.   This paradigm shift has led to the implementation of \textit{machine learning-based routing}, where models learn from traffic patterns, topology changes, and link STATEs to optimize end-to-end performance \cite{mehrabi2021survey, oymak2019overparameterized}.
As learning-driven processes grow more widespread, guaranteeing \textit{fairness across various nodes} in terms of latency, throughput, and energy consumption has arisen as an important yet unsolvable challenge.   Resource-rich nodes often dominate learning and routing decisions, putting weaker or remote devices behind \cite{barocas2023fairness}.   Current optimization frameworks frequently use \textit{Euclidean space}, which does not accurately represent the hierarchical and curved nature of real-world network topologies \cite{chen2024toward}. These limitations highlight the need for learning frameworks that are geometry-aware and fair, enabling equitable and efficient routing across \textcolor{black}{decentralized edge-level learning under a federated aggregation framework.}

Traditional routing systems, such as OSPF, BGP, and reinforcement-learning extensions, focus on network-wide efficiency, aiming to minimize latency or improve throughput \cite{bogyrbayeva2024machine, liu2021deep}.   These techniques provide excellent aggregate performance, but favor well-connected or resource-rich nodes, resulting in \textit{unfair service allocation} in heterogeneous edge environments \cite{wang2023survey}.  Recent research has offered federated routing and distributed reinforcement-learning frameworks, which allow for local policy modifications without centralizing traffic data, enhancing scalability and privacy \cite{cong2023soho, zhang2024unified, liu2024federated}. However, these algorithms presume statistically homogeneous clients and minimize variances in node capacity, connection, and energy budgets, resulting in \textcolor{black}{non-independent and identically distributed (non-IID) learning behavior
} and biased routing decisions.  In parallel, \textit{graph neural-network-based} routing models have emerged to take use of network topologies' structural relationships \cite{xu2024scalable, guang2024graph}.  However, the bulk of these models operate on \textit{Euclidean embedding spaces}, which fail to reflect the hierarchical or power-law connection patterns found in real-world networks \cite{chen2024toward}.  As a result, they are unable to establish geometric connections, which are required for balancing traffic loads across distant or hierarchically organized sub-networks.   These challenges underscore the necessity for a paradigm that tackles both \textit{topological geometry and fairness} in a federated, large-scale routing scenario.

\textbf{Main Contributions:} To address such challenges, this paper offers \textit{Geo-FairFed}, a federated routing architecture for large-scale edge networks that considers geometric fairness.   The main objective is to embed the network topology in a \textit{hyperbolic manifold}, allowing each node to learn structural connections that better represent hierarchical and asymmetric connectivity.   In this geometric space, each edge device trains a routing policy with a \textit{hyperbolic graph neural network (HGNN)}, while a global federated aggregator combines model updates using a \textit{fairness-constrained objective} that balances routing efficiency and equity \cite{xu2024scalable}.   Geo-FairFed enables equitable routing decisions among heterogeneous nodes by integrating geometric representation learning and fairness-aware optimization without affecting overall network performance. The main contributions of this paper are summarized as follows.

\begin{itemize}
    \item We introduce a \textit{hyperbolic graph neural network (HGNN)} to learn topology-aware representations of network nodes and links.   This enables the efficient encoding of hierarchical and asymmetric patterns commonly encountered in large-scale edge networks.
    \item We develop a \textit{federated routing method} that incorporates a fairness criteria throughout the aggregation process.   The design balances routing performance among varied edge devices by including Jain's fairness index \cite{rezaeinia2023efficiency} into the global objective.
    \item We propose a \textit{unified optimization strategy} that minimizes routing delay, geometric inconsistency, and injustice.   We also give a theoretical study that shows convergence under limited curvature and fairness-regularized stability.
    \item We conduct comprehensive simulations on dynamic 6G and IoT topologies, and the proposed Geo-FairFed framework delivers reduced latency (20\%), energy consumption (17\%), and higher fairness than state-of-the-art federated and geometric routing benchmarks.
\end{itemize}

\section{Related Works}

\subsection{Fairness in Networking}
Fairness has long been a key factor in network design, ensuring that resources such as bandwidth, transmission power, and latency are properly distributed across users.   Early research offered the max-min and proportional fairness concepts for wired and wireless networks, which tried to balance throughput and latency under shared connection constraints \cite{wang2023survey}.   The Jain's fairness index has become a widely accepted quantitative method for analyzing network performance equality across several nodes \cite{rezaeinia2023efficiency}.   Subsequent attempts extended these ideas to include QoS-oriented fairness and energy-efficient fairness, with the goal of preventing network shortage or bias toward high-capacity nodes \cite{liu2021deep}.
Recently, fairness has been included into reinforcement-learning-based routing and congestion control, with rules learned to equalize queueing delays or packet delivery rates \cite{zhang2024unified}.   However, these systems usually rely on centralized training and fail to scale in heterogeneous edge environments where nodes differ in computation, mobility, and link quality.  

\subsection{Federated Optimization and Edge Routing}
Federated learning (FL) is a scalable approach for distributed model training over heterogeneous edge devices, enabling local computation without transferring raw data \cite{cong2023soho}.  In networking, FL has been researched for adaptive routing, traffic prediction, and congestion management. Each node or cluster trains a local policy based on its own traffic statistics and frequently contributes to a global model \cite{zhang2024unified}. However, these strategies are primarily concerned with global accuracy and convergence, ignoring the fairness of participating nodes.
Clients on large-scale edge and 6G networks usually differ in bandwidth, mobility, and computing capability, leading in biased aggregation in which powerful nodes have a disproportionate influence on global routing strategy.   Furthermore, federated routing models frequently operate on Euclidean feature spaces, ignoring network topologies' underlying geometric links and hierarchical organization \cite{chen2024toward}.   Several studies have explored fairness-aware aggregation algorithms in FL for applications such as categorization or healthcare \cite{zhang2024unified}; however, they have not been extended to network-level fairness, where routing decisions have a direct impact on link use and latency distribution \cite{liu2024federated}.   As a result, traditional federated optimization algorithms are insufficient for providing topology-aware and equitable routing performance in a variety of edge configurations.

\subsection{Geometric and Graph Learning for Network Control}
Graph-based learning has recently emerged as a significant tool for modeling networked systems, as it combines spatial and relational links that standard neural networks ignore \cite{guang2024graph}.   Graph Neural Networks (GNNs) have been used for routing, topology prediction, and traffic engineering, with nodes sharing hidden STATEs to represent connection patterns.   These models have demonstrated impressive generalization to previously unknown topologies and dynamic conditions, outperforming rule-based routing protocols \cite{xu2024scalable}.   However, most GNNs use Euclidean embeddings, which assume flat geometry and fail to accurately represent the hierarchical or power-law patterns inherent in large-scale communication networks.
To overcome this issue, geometric deep learning has proposed non-Euclidean embeddings, such as hyperbolic and spherical spaces, that better represent hierarchical, clustered, and asymmetric topologies.   Hyperbolic graph embeddings have shown superior performance in modeling Internet-scale networks, knowledge graphs, and wireless topologies with latent hierarchies \cite{chen2024toward}.   Despite these benefits, few research have looked into incorporating geometric representation learning into routing algorithms, and even fewer have explored its relationship with federated optimization or fairness constraints to balance routing efficiency or promote fairness among various nodes. 

\textbf{Research Gap.}
In summary, past research 
have largely evolved in isolation.   Fairness approaches in networking tend to be centered on static allocation or centralized scheduling, with minimal flexibility for changing topologies. Federated optimization approaches enable scalability and anonymity, but they frequently fail to ensure equitable involvement across heterogeneous nodes.
Meanwhile, geometric and graph-based learning accurately represent network hierarchies, but they do not account for fairness or aggregation bias in distributed training.   As a result, no single approach guarantees topological awareness, federated scalability, and equitable performance distribution.  

\section{System Model and Problem Formulation}

\subsection{Network Model and Notation}
\textcolor{black}{Consider a communication network as a weighted graph $G = (V, \mathcal{E})$, where $V = \{v_1, v_2, \ldots, v_N\}$ represents the set of $N$ edge nodes and $\mathcal{E} \subseteq V \times V$ represents the set of communication connections. 
 Each connection $(i,j) \in \mathcal{E}$ is identified by its bandwidth $b_{ij}$, propagation delay $d_{ij}$, and transmission energy cost $e_{ij}$. }
 The adjacency matrix $A \in \mathbb{R}^{N \times N}$ represents node connectivity. Each node $v_i$ has a feature vector $\mathbf{x}_i \in \mathbb{R}^m$ that describes its local state (e.g., queue length, packet arrival rate, link usage, or channel condition).

{\color{black} Each $v_i \in V$ represents a network edge node, which is identified by a local feature vector $x_i$, develops a local routing strategy $f_i(\mathbf{x}_i; \theta_i)$ that maps their neighborhood STATE to a next-hop selection or forwarding probability distribution. All nodes initialize their local parameters from the same global model at the start of training, i.e., $\theta_i^{(0)} = \Theta^{(0)}$ for all $i$. Each node then updates $\theta_i^{(t)}$ locally and helps form the global parameters $\Theta^{(t+1)}$ at each communication round.} 
 Every $T$ communication rounds, the nodes exchange data with a central aggregator (for example, an edge server). 
 To express hierarchical linkages and heterogeneous connectivity, node embeddings are represented in a negatively curved manifold $\mathbb{H}^d$ with a curvature value $\kappa < 0$. 
 Let $\mathbf{h}_i = \phi(\mathbf{x}_i; \theta_i) \in \mathbb{H}^d$ represent the embedding of node $v_i$ derived from an HGNN encoder $\phi(\cdot)$.

\subsection{Local Learning Objective}
Every node locally trains its routing model by minimizing a composite loss function as
\begin{equation}
\mathcal{L}_i^{\text{local}} 
= \alpha \, \mathcal{L}_i^{\text{delay}} 
+ \beta \, \mathcal{L}_i^{\text{energy}} 
+ \gamma \, \mathcal{L}_i^{\text{drop}},
\end{equation}
where $\mathcal{L}_i^{\text{delay}}$ represents the average end-to-end latency for flows passing through node. The variables $v_i$, $\mathcal{L}_i^{\text{energy}}$, and $\mathcal{L}_i^{\text{drop}}$ represent the energy consumed per packet and the local packet loss rate. 
The coefficients $\alpha, \beta, \gamma \geq 0$ balance the objectives based on network priority and traffic conditions. \textcolor{black}{Although the network topology remains constant over short time intervals, local losses vary between epochs due to stochastic traffic arrivals, evolving queue states, and time-varying link utilization within each observation window, requiring multiple local optimization epochs over the collected mini-batch of routing samples.}

\subsection{Fairness Definition}
To ensure equitable routing performance among diverse nodes, we use Jain's fairness index \cite{rezaeinia2023efficiency}, which is defined as
\begin{equation}
\mathcal{F}(\boldsymbol{\eta}) 
= \frac{\left(\sum_{i=1}^{N} \eta_i \right)^2}
{N \sum_{i=1}^{N} \eta_i^2},
\end{equation}
where $\eta_i$ represents a performance statistic for node $v_i$, such throughput or inverse latency. 
 The fairness loss is defined as $\mathcal{L}_{\text{fair}} = 1 - \mathcal{F}(\boldsymbol{\eta})$. Lower values indicate more balanced performance across nodes. \textcolor{black}{The proposed structure is not restricted to any particular measure, however this study uses Jain's fairness index because of its analytical ease and smoothness. By changing the fairness loss term in the global goal, other fairness concepts like max--min fairness, proportional fairness, or application-level QoE-based fairness may be included. Because of its adaptability, Geo-FairFed may support various fairness goals based on service-level priorities and deployment requirements.}

\subsection{Global Optimization Problem}
The objective is to reduce routing costs, unfairness, and geometric distortion for all participating nodes. 
 Formally, the global objective is provided as
\begin{equation}
\min_{\Theta}
\left[
\mathcal{L}_{\text{routing}}(\Theta)
+ \lambda_1 \mathcal{L}_{\text{fair}}(\Theta)
+ \lambda_2 \mathcal{L}_{\text{geo}}(\Theta)
\right],
\label{eq:global}
\end{equation}
where $\theta = \{\theta_1, \ldots, \theta_N\}. $ denotes all model parameters: $\mathcal{L}_{\text{routing}} = \frac{1}{N} \sum_{i=1}^{N} \mathcal{L}_i^{\text{local}}$ represents the average routing loss, $\mathcal{L}_{\text{fair}}$ penalizes inequality across nodes, and $\mathcal{L}_{\text{geo}}$ ensures manifold consistency to maintain curvature information across local and global embeddings. 
 The trade-off coefficients $\lambda_1, \lambda_2 > 0$ determine the strength of the fairness and geometric regularization terms, respectively.

\begin{figure}
    \centering
    \includegraphics[width=0.99\linewidth]{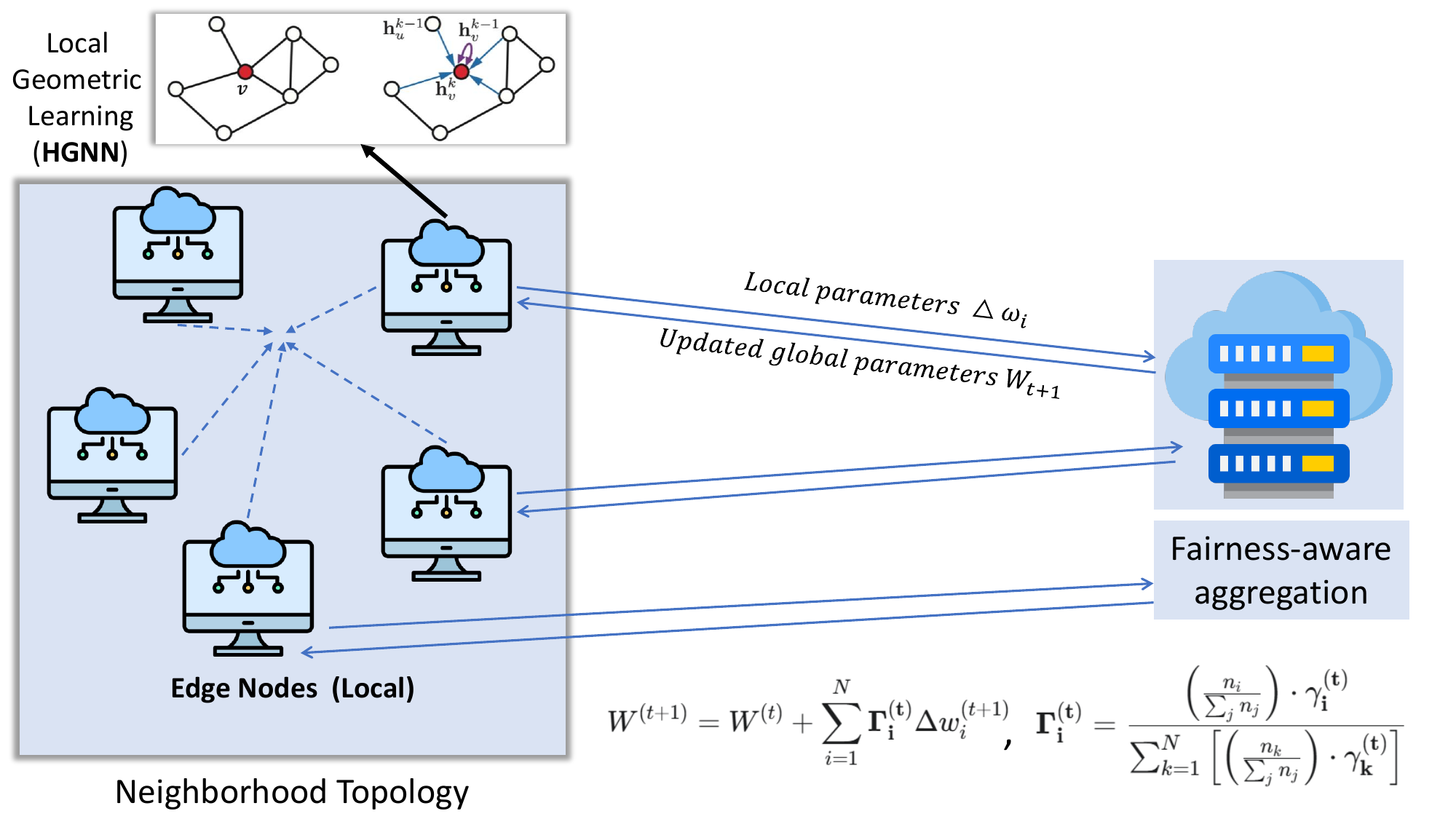}
    \caption{Overview of the proposed Geo-FairFed framework, including HGNN, model update, and fairness-aware aggregation.}
    \label{fig:architecture}
\end{figure}

Each node's routing decision is further limited with
\begin{align}
\color{black}
0 \leq p_i \leq p_i^{\max}, \quad 
d_{ij} \leq d_{\max}, \quad 
E_i \leq E_i^{\max}, \quad \forall (i,j) \in \mathcal{E},
\end{align}
where $p_i$ represents transmission power, $E_i$ denotes the local energy
consumption, and $d_{\max}$ is the maximum permitted delay per connection.


\section{Proposed Geo-FairFed Framework}

\subsection{Overall Framework Overview}
As shown in Fig.~\ref{fig:architecture}, Geo-FairFed operates in iterative communication rounds with three primary stages:
(i) \textbf{Local geometric routing learning}, involves each edge node training a hyperbolic graph neural network (HGNN) based on its neighborhood topology and traffic information;
(ii) \textbf{Model update exchange}, involves nodes sending just parameter updates to the central aggregator, without sharing raw data; and
(iii) \textbf{Fairness-aware aggregation}, involves updating the global model using a weighted combination of local updates based on fairness coefficients. 


\subsection{Geometric Routing Embedding}
Each node represents its local neighborhood in a latent representation in a negatively curved hyperbolic space $\mathbb{H}^d$. Let $\mathbf{h}_i^{(l)} \in \mathbb{H}^d$ represent the embedding of node $v_i$ at layer $l$. Message passing in the HGNN layer operates as
\begin{equation}
\mathbf{h}_i^{(l+1)} 
= \text{MLP}^{\mathbb{H}} \left(
\text{Agg}^{\mathbb{H}}\left(
\{\exp_0^{-1}(\mathbf{h}_j^{(l)}) : j \in \mathcal{N}(i)\}
\right)
\right),
\end{equation}
where $\exp_0^{-1}(\cdot)$ maps points from the manifold to the tangent space at the origin, whereas $\text{Agg}^{\mathbb{H}}(\cdot) $ represents a hyperbolic aggregation operator (e.g., Möbius addition).
 The embedding geometry is controlled by the curvature parameter $\kappa < 0$, with greater values capturing stronger hierarchical links. 
 The local geometric loss penalizes deviations from manifold consistency as
\begin{equation}
\mathcal{L}_{\text{geo}} = \sum_{(i,j) \in \mathcal{E}}
\left|  d_{\mathbb{H}}(\mathbf{h}_i, \mathbf{h}_j) 
- d_{ij} \right|,
\end{equation}
where $d_{\mathbb{H}}(\cdot,\cdot)$ denotes the geodesic distance in $\mathbb{H}^d$, and $d_{ij}$ represents the observed link latency.

\subsection{Fairness-Aware Federated Aggregation}
After each local training phase, the model parameters $\theta_i$ are submitted to the central aggregator. 
 A standard FedAvg update combines local models as
\begin{equation}
\theta^{(t+1)} = \sum_{i=1}^{N} w_i \, \theta_i^{(t)},
\label{eq:fedavg}
\end{equation}
where $w_i = n_i / \sum_j n_j$, where $n_i$ represents the number of samples at node $i$. 
 Geo-FairFed promotes fairness by replacing $w_i$ with a fairness-adjusted weight $\tilde{w}_i$ described as
\begin{equation}
\tilde{w}_i = 
\frac{(1 - \rho_i) \, n_i}
{\sum_j (1 - \rho_j) \, n_j},
\quad 
\rho_i = 1 - \mathcal{F}(\eta_i),
\end{equation}
where $\mathcal{F}(\eta_i)$ is the local Jain's fairness index for node~$i$.
 Nodes with higher inequality (bigger $\rho_i$) receive lower aggregate weights, preventing dominating nodes from biasing the global model.

\subsection{Joint Optimization Objective}
The global training objective combines routing efficiency, fairness, and geometric consistency as
\begin{equation}
\min_{\Theta} \mathcal{L}_{\text{routing}} 
+ \lambda_1 \, \mathcal{L}_{\text{fair}} 
+ \lambda_2 \, \mathcal{L}_{\text{geo}},
\label{eq:joint}
\end{equation}
where the mean local routing loss is denoted by $\mathcal{L}_{\text{routing}}$, and $\mathcal{L}_{\text{fair}} = 1 - \mathcal{F}(\boldsymbol{\eta})$ and $\mathcal{L}_{\text{geo}}$ are defined above. The parameters $\lambda_1$ and $\lambda_2$ determine the balance between fairness and geometric regularization.

\subsection{Algorithm Description}
\textcolor{black}{The training process is summarized in Algorithm~\ref{alg:geofairfed}. 
 Edge nodes perform $E_{\text{loc}}$ local epochs of HGNN-based training, followed by federated aggregation using fairness-adjusted weights.
 Geo-FairFed starts by setting up global model parameters \eqref{algline:init}.  
 Each communication round, each client computes \emph{hyperbolic embeddings} of its local neighborhood \eqref{algline:embed} and executes $E_{\text{loc}}$ local optimization epochs, minimizing the routing and geometry-aware loss \eqref{algline:localupdate}.}
After local training, the client measures its performance metric $\eta_i$ and generates a fairness score $\rho_i$ based on Jain's index \eqref{algline:fairscore}. This score, along with model parameters and sample counts, is uploaded to the server \eqref{algline:upload}.  
 The server computes fairness-adjusted aggregation weights $\tilde{w}_i$ \eqref{algline:weights}, creates the new global model using a weighted average \eqref{algline:aggregate}, and broadcasts the updated parameters to all participants \eqref{algline:broadcast}.  
 These processes are repeated for $T$ rounds until we get a fairness-aware equilibrium routing policy. \textcolor{black}{In this case, the parameter $T$ indicates the length of the local traffic observation window used to gather routing samples, and the index $t$ indicates a federated communication round rather than a physical time unit.}

\begin{algorithm}[t]
\caption{Geo-FairFed: Geometric Fairness-Aware Federated Routing}
\label{alg:geofairfed}
\small
\begin{algorithmic}[1]
\STATE \textbf{Input:} clients $V$, communication rounds $T$, local epochs $E_{\text{loc}}$, curvature $\kappa<0$, coefficients $\lambda_1,\lambda_2$, learning rate $\eta$
\STATE \textbf{Initialize:} global parameters $\Theta^{(0)}$; set $\theta_i^{(0)} = \Theta^{(0)}$ for all $i$ \label{algline:init}
\FOR{$t = 0,1,\dots,T-1$}
  \FORALL{client $i \in V$ \textbf{in parallel}}
    \STATE \textbf{Hyperbolic Encoding:} compute HGNN embeddings $\mathbf{h}_i  \leftarrow  \phi(\mathbf{x}_i;\theta_i^{(t)},\kappa)$ \label{algline:embed}
    \FOR{$e = 1$ \TO $E_{\text{loc}}$}
        \STATE \textbf{Local Update:} $\theta_i^{(t)}  \leftarrow  \theta_i^{(t)} - \eta \nabla_{\theta} \big(\mathcal{L}_i^{\text{local}} + \lambda_2 \mathcal{L}_{\text{geo}}\big)$ \label{algline:localupdate}
    \ENDFOR
    \STATE \textbf{Compute Fairness:} evaluate node metric $\eta_i$ and local fairness index $\mathcal{F}(\eta_i)$; set $\rho_i  \leftarrow  1 - \mathcal{F}(\eta_i)$ \label{algline:fairscore}
    \STATE \textbf{Upload:} send $(\theta_i^{(t)}, \rho_i, n_i)$ to the server \label{algline:upload}
  \ENDFOR
  \STATE \textbf{Fairness-Adjusted Weights:} $\tilde{w}_i  \leftarrow  \dfrac{(1-\rho_i)\,n_i}{\sum_{j\in V}(1-\rho_j)\,n_j}$ \label{algline:weights}
  \STATE \textbf{Aggregation:} $\theta^{(t+1)}  \leftarrow  \sum_{i\in V} \tilde{w}_i\,\theta_i^{(t)}$ \label{algline:aggregate}
  \STATE \textbf{Broadcast:} send $\theta^{(t+1)}$ to all clients \label{algline:broadcast}
\ENDFOR
\STATE \textbf{Output:} fairness-aware global routing model $\theta^{(T)}$
\end{algorithmic}
\end{algorithm}

\subsection{Complexity Analysis}
The computational complexity for each client is $O(|E_i| d L)$, where $|E_i|$ is the number of local edges, $d$ is the embedding size, and $L$ is the number of HGNN layers. The communication cost per round is $O(|\Theta| p)$, where $p$ represents the client participation percentage and $|\Theta|$ is the model size. Fairness regularization incurs a negligible $O(N)$ overhead while computing $\rho_i$ values.

\textbf{Summary.}
Geo-FairFed enhances routing efficiency, fairness, and geometric consistency via distributed HGNN training and fairness-aware aggregation. 
 The next part outlines the experimental setup and assessment criteria used to evaluate performance in realistic 6G and IoT network scenarios.
\textcolor{black}{It is possible to analytically define the absolute communication overhead per round in addition to normalized complexity. Let $N$ be the number of participating nodes and $d$ be the number of model parameters. A communication cost of about $2Nd$ scalar values is incurred every round when each node sends its local model parameters to the aggregator and gets the updated global model. Bytes each communication round, assuming 32-bit floating-point format, are $8Nd$. This approach enhances the normalized scalability trends in Fig.~6 and enables direct estimate of communication overhead in real-world deployments.}

\section{Theoretical Analysis}


\begin{assumption}[Smoothness]
\label{asm:smoothness}
The local loss functions $\mathcal{L}_i^{\text{local}}(\theta_i)$ are $L$-smooth, meaning that for each $\theta_i, \theta_i'$,
\begin{equation}
||\nabla \mathcal{L}_i^{\text{local}}(\theta_i) - \nabla \mathcal{L}_i^{\text{local}}(\theta_i')||
\leq L ||\theta_i - \theta_i'||.
\end{equation}
This ensures that the gradient does not abruptly change, resulting in stability in local optimization.
\end{assumption}

\begin{assumption}[Bounded Gradient Variance]
\label{asm:variance}
The stochastic gradients used for local updates have bounded variance, i.e.,
\begin{equation}
\mathbb{E}||\nabla \mathcal{L}_i^{\text{local}} - \nabla \mathcal{L}_{\text{routing}}||^2
\leq \sigma^2, \quad \forall i.
\end{equation}
This is a standard assumption in federated optimization that accommodates for data heterogeneity.
\end{assumption}

\begin{assumption}[Bounded Curvature]
\label{asm:curvature}
The hyperbolic manifold $\mathbb{H}^d$ has a curvature parameter $\kappa$ with $|\kappa| < \kappa_{\max}$. The exponential and logarithmic mappings $\exp_0(\cdot)$ and $\exp_0^{-1}(\cdot)$ are bi-Lipschitz continuous zx
\begin{equation}
c_1 ||u-v|| \le d_{\mathbb{H}}(\exp_0(u), \exp_0(v)) \le c_2 ||u-v||,
\end{equation}
where $d_{\mathbb{H}}(\cdot,\cdot)$ indicates the geodesic distance.  
 This ensures that the curvature does not skew gradients significantly.
\end{assumption}

\begin{assumption}[Fairness Regularity]
\label{asm:fairness}
The fairness loss $\mathcal{L}_{\text{fair}}(\Theta)$ is convex and continuously differentiable with respect to $\Theta$. Its gradient is bounded by a constant $C_{\text{fair}}$, i.e., $||\nabla \mathcal{L}_{\text{fair}}|| \le C_{\text{fair}}$.  
 This guarantees that the fairness constraint affects the optimization process smoothly.
\end{assumption}


\begin{remark}
Assumption~\ref{asm:curvature} is relatively small because most real-world network topologies have weakly negative curvature ($|\kappa| < 1$).  
 Assumption~\ref{asm:fairness} is consistent with empirical fairness penalties that are differentiable and Lipschitz-bounded, making these criteria applicable to large-scale edge deployments.
\end{remark}

{\color{black}
\noindent\textbf{Practical considerations.}
Convergence guarantees in federated and geometric optimization are established using the basic analytical requirements mentioned above. Although non-stationary traffic, topological changes, or adversarial behavior may be present in real-world edge environments, these assumptions are mostly used to facilitate theoretical analysis rather than to limit practical deployment. The proposed structure exhibits resilience beyond idealized analytical circumstances, as evidenced by the empirical assessment (Section VI), which shows that it retains stable performance under realistic network variability.
}

\subsection{Convergence of Local Training}
\textcolor{black}{We investigate the descent behavior of a single client $i$ over $E_{\text{loc}}$ local epochs in round $t$.} Let $\theta_i^{(t,e)}$ represent the local parameters after $e \in \{0,\dots,E_{\text{loc}}\}$ epochs inside round $t$ (with $\theta_i^{(t,0)} = \theta^{(t)}$ the received global model). To clarify, we write $\mathcal{L}_i \equiv \mathcal{L}_i^{\text{local}} + \lambda_2 \mathcal{L}_{\text{geo}}$ as the local objective minimized at the client (see Line~\eqref{algline:localupdate}). 
 Local updates in the tangent space are achieved using logarithmic/exponential maps with curvature $\kappa < 0$:
\[
\tilde{\theta}_i^{(t,e+1)}  =  \theta_i^{(t,e)}  -  \eta\, g_i^{(t,e)}, 
\quad 
\theta_i^{(t,e+1)}  =  \exp_0 \big(\,\tilde{\theta}_i^{(t,e+1)}\,\big),
\]
where $g_i^{(t,e)}$ is an unbiased stochastic gradient estimator of $\nabla \mathcal{L}_i(\theta_i^{(t,e)})$ calculated in the tangent space at the origin (using parallel transport if required).

\begin{lemma}[One-Step Local Descent]
\label{lem:local_descent}
Let Assumptions~\ref{asm:smoothness}--\ref{asm:curvature} hold and let $\eta \in (0,\frac{1}{L}]. $. 
 Assume $\mathbb{E}[g_i^{(t,e)}] = \nabla \mathcal{L}_i(\theta_i^{(t,e)})$ and $\mathbb{E}||g_i^{(t,e)} - \nabla \mathcal{L}_i(\theta_i^{(t,e)})||^2 \le \sigma^2$. 
 The predicted local objective matches
\begin{equation}
\label{eq:one_step_descent}
\begin{aligned}
    &\mathbb{E}\big[\mathcal{L}_i(\theta_i^{(t,e+1)})\big]
 \le 
\\&\mathbb{E}\big[\mathcal{L}_i(\theta_i^{(t,e)})\big]
- \frac{\eta}{2}\,
\mathbb{E}\big[||\nabla \mathcal{L}_i(\theta_i^{(t,e)})||^2\big]
+ \frac{L\eta^2}{2}\,\sigma^2
+ C_{\kappa}\,\eta^2,
\end{aligned}
\end{equation}
where $C_{\kappa} = \tfrac{L}{2}(c_2^2 - 1)\,G^2$ is a curvature-induced constant that depends on the bi-Lipschitz factor $c_2$ in Assumption~\ref{asm:curvature} and an upper bound $G$ on gradient norms.
\end{lemma}

\begin{proof}
The $L$-smoothness of $\mathcal{L}_i$ in Assumption~\ref{asm:smoothness} implies that for any vectors $u,v$ in the \emph{same} space,
\[
\mathcal{L}_i(v) \le \mathcal{L}_i(u) 
+ \langle \nabla \mathcal{L}_i(u), v-u \rangle 
+ \frac{L}{2}||v-u||^2.
\]
We use this inequality in the \emph{tangent} space around the present point. 
 Let $u = \theta_i^{(t,e)}$ (identified with $\log_0(\theta_i^{(t,e)})$) and determine the trial point $\tilde{v} = \tilde{\theta}_i^{(t,e+1)} = \theta_i^{(t,e)} - \eta g_i^{(t,e)}$. 
 We first get a constraint for the tangent-space update, and then use $\exp_0(\cdot)$ to regulate the distortion when mapping back to the manifold.

\textbf{(i) Tangent step.)} By smoothness,
\begin{equation}
    \begin{aligned}
        &\mathcal{L}_i(\tilde{\theta}_i^{(t,e+1)}) 
\le \\&\mathcal{L}_i(\theta_i^{(t,e)}) 
+ \left\langle \nabla \mathcal{L}_i(\theta_i^{(t,e)}), -\eta\, g_i^{(t,e)} \right\rangle
+ \frac{L}{2}\,\eta^2 ||g_i^{(t,e)}||^2.
    \end{aligned}
\end{equation}
Using $\mathbb{E}[g_i^{(t,e)}]$ to compute conditional expectation with respect to local randomness at epoch $\mathbb{E}[g_i^{(t,e)}] = \nabla \mathcal{L}_i(\theta_i^{(t,e)})$, we get
\begin{equation}
    \begin{aligned}
        &\mathbb{E} \left[\mathcal{L}_i(\tilde{\theta}_i^{(t,e+1)}) \mid \theta_i^{(t,e)}\right]
\le \\&\mathcal{L}_i(\theta_i^{(t,e)}) 
- \eta ||\nabla \mathcal{L}_i(\theta_i^{(t,e)})||^2
+ \frac{L}{2}\eta^2\,\mathbb{E}||g_i^{(t,e)}||^2.
    \end{aligned}
\end{equation}
Breaking down $\mathbb{E}||g_i^{(t,e)}||^2 = ||\nabla \mathcal{L}_i(\theta_i^{(t,e)})||^2 + \mathbb{E}||g_i^{(t,e)} - \nabla \mathcal{L}_i||^2 \le ||\nabla \mathcal{L}_i||^2 + \sigma^2$, we obtain
\begin{equation}
    \begin{aligned}
        &\mathbb{E} \left[\mathcal{L}_i(\tilde{\theta}_i^{(t,e+1)}) \mid \theta_i^{(t,e)}\right]
\le \\&\mathcal{L}_i(\theta_i^{(t,e)}) 
- \eta \Big(1 - \tfrac{L\eta}{2}\Big)||\nabla \mathcal{L}_i||^2
+ \frac{L}{2}\eta^2 \sigma^2.
    \end{aligned}
\end{equation}
Given that $\eta\le \tfrac{1}{L}$, we have $1 - \tfrac{L\eta}{2}\ge \tfrac{1}{2}$. Therefore,
\begin{equation}
\label{eq:tangent_decrease}
\begin{aligned}
    &\mathbb{E} \left[\mathcal{L}_i(\tilde{\theta}_i^{(t,e+1)}) \mid \theta_i^{(t,e)}\right]
\le \\&\mathcal{L}_i(\theta_i^{(t,e)}) 
- \frac{\eta}{2}||\nabla \mathcal{L}_i(\theta_i^{(t,e)})||^2
+ \frac{L}{2}\eta^2 \sigma^2.
\end{aligned}
\end{equation}

\textbf{(ii) Mapping back to the manifold.)} 
We now associate $\mathcal{L}_i(\theta_i^{(t,e+1)}) = \mathcal{L}_i(\exp_0(\tilde{\theta}_i^{(t,e+1)})). $to $\mathcal{L}_i(\tilde{\theta}_i^{(t,e+1)})$. 
 The exponential map's bi-Lipschitz property (Assumption~\ref{asm:curvature}) limits fluctuations measured in the manifold to no more than $c_2$ times those in the tangent space. 
 Using typical smoothness arguments with retraction maps (considering $\exp_0$ as a retraction), the incurred distortion per step can be bounded by
\begin{equation}
    \Delta_{\kappa} 
 \le  \frac{L}{2}\,(c_2^2 - 1)\,||\tilde{\theta}_i^{(t,e+1)} - \theta_i^{(t,e)}||^2 
= \frac{L}{2}\,(c_2^2 - 1)\,\eta^2 ||g_i^{(t,e)}||^2.
\end{equation}
Using expectation and $\mathbb{E}||g_i^{(t,e)}||^2 \le G^2$, we get $\mathbb{E}[\Delta_{\kappa}] \le C_{\kappa}\eta^2$, where $C_{\kappa} = \tfrac{L}{2}(c_2^2-1)G^2$.

Combining \eqref{eq:tangent_decrease} with the curvature term produces \eqref{eq:one_step_descent} after eliminating the conditioning and taking total expectation.
\end{proof}

\begin{corollary}[Multi-Epoch Local Descent]
\label{cor:multi_epoch_descent}
Based on Lemma~\ref{lem:local_descent}, for $E_{\text{loc}}$ local epochs, we have
\begin{equation}
    \begin{aligned}
        \mathbb{E}\big[\mathcal{L}_i(\theta_i^{(t,E_{\text{loc}})})\big]
\le  &\mathbb{E}\big[\mathcal{L}_i(\theta_i^{(t,0)})\big]
- \frac{\eta}{2}\sum_{e=0}^{E_{\text{loc}}-1}\mathbb{E}\big[||\nabla \mathcal{L}_i(\theta_i^{(t,e)})||^2\big]
\\&+ E_{\text{loc}}\Big(\frac{L\eta^2}{2}\sigma^2 + C_{\kappa}\eta^2\Big).
    \end{aligned}
\end{equation}
\end{corollary}

\begin{remark}[Interpretation]
The inequality \eqref{eq:one_step_descent} shows a conventional SGD decrease term $-\tfrac{\eta}{2}||\nabla \mathcal{L}_i||^2$ plus two penalties: the stochastic variance term $\tfrac{L\eta^2}{2}\sigma^2$ and a curvature-induced term $C_{\kappa}\eta^2$. 
 Smaller $\eta$ and substantial negative curvature (small $c_2^2 - 1$) negate the curvature penalty, indicating robust HGNN training on $\mathbb{H}^d$.
\end{remark}

\subsection{Convergence of Fairness-Weighted Federated Aggregation}
We investigate the global update made at the server after $E_{\text{loc}}$ local epochs. 
Let $\theta_i^{(t)} \equiv \theta_i^{(t,E_{\text{loc}})}$ be client $i$'s model at the conclusion of round $t$. Define the fairness-weighted aggregation
\(
\theta^{(t+1)}  =  \sum_{i=1}^N \tilde{w}_i^{(t)} \theta_i^{(t)}
\)
with
\(
\tilde{w}_i^{(t)} 
= \frac{(1-\rho_i^{(t)}) n_i}{\sum_{j}(1-\rho_j^{(t)}) n_j}
\),
where $n_i$ represents the local sample count and $\rho_i^{(t)} \in [0,1]$ is the local inequality score (Line~\eqref{algline:fairscore}).

\begin{lemma}[Weight Normalization and Bounded Bias]
\label{lem:weights}
Let $0 \le \rho_i^{(t)} \le \rho_{\max} < 1$ for all $i,t$ and $n_i \in \mathbb{N}_+$, subsequently we get
\begin{enumerate}
\item \textbf{Normalization:} $\sum_{i=1}^N \tilde{w}_i^{(t)} = 1$ and $\tilde{w}_i^{(t)}  \ge  0$. 
\item \textbf{Deviation from FedAvg:} Let the FedAvg weights be $w_i^{\text{FA}} = \frac{n_i}{\sum_j n_j}$. Following that, we obtain
\[
 ||\tilde{\mathbf{w}}^{(t)} - \mathbf{w}^{\text{FA}} ||_1 
\le \frac{2\rho_{\max}}{1-\rho_{\max}},
\]
i.e., the fairness-adjusted mixture maintains a bounded $\ell_1$ distance from FedAvg.
\item \textbf{Bounded Aggregation Drift:} If $ ||\theta_i^{(t)}-\theta^{(t)} || \le \Delta_t$ for all $i$, then
\[
\Big ||
\sum_i \tilde{w}_i^{(t)} \theta_i^{(t)} 
- \sum_i w_i^{\text{FA}} \theta_i^{(t)}
\Big || 
\le \Delta_t \, \frac{2\rho_{\max}}{1-\rho_{\max}}.
\]
\end{enumerate}
\end{lemma}

\begin{proof}
(1) Nonnegativity is immediate because $(1-\rho_i^{(t)}) n_i \ge 0$.  Summing the numerator over $i$ equals the denominator, resulting in normalization.

(2) Consider $a_i=(1-\rho_i^{(t)})n_i$, $A=\sum_j a_j$, and $N=\sum_j n_j$.  Then, $\tilde{w}_i^{(t)}=\frac{a_i}{A}$, and $w_i^{\text{FA}}=\frac{n_i}{N}$. 
 Because $0 \le \rho_i^{(t)} \le \rho_{\max}$, we have $(1-\rho_{\max}) n_i \le a_i \le n_i$, therefore $(1-\rho_{\max})N \le A \le N$. 
 Applying the conventional ratio-difference bound as
\begin{equation}
    \begin{aligned}
        \Big|\frac{a_i}{A}-\frac{n_i}{N}\Big|
&\le 
\frac{|a_iN - n_iA|}{A N}
\le 
\frac{n_i(N-A) + (n_i-a_i)N}{A N}
\\&\le 
\frac{n_i \rho_{\max} N + n_i \rho_{\max} N}{(1-\rho_{\max})N^2}
= \frac{2\rho_{\max}}{(1-\rho_{\max})N} n_i.
    \end{aligned}
\end{equation}
Summing over $i$ produces the $\ell_1$ bound.

(3) By convexity,
\begin{equation}
    \begin{aligned}
        \Big ||
\sum_i (\tilde{w}_i^{(t)}-w_i^{\text{FA}})\theta_i^{(t)}
\Big ||
&\le 
\sum_i |\tilde{w}_i^{(t)}-w_i^{\text{FA}}|\, ||\theta_i^{(t)}-\theta^{(t)} ||
\\&\le \Delta_t  ||\tilde{\mathbf{w}}^{(t)}-\mathbf{w}^{\text{FA}} ||_1,
    \end{aligned}
\end{equation}
and apply item (2).
\end{proof}

\begin{lemma}[Fairness-Weighted Gradient Consistency]
\label{lem:grad_consistency}
Consider $\nabla \mathcal{L}_i(\theta)$ denote the local gradients for the joint local objective 
$\mathcal{L}_i^{\text{local}}+\lambda_2\mathcal{L}_{\text{geo}}$, and 
let $\bar{g}^{\text{FA}}=\sum_i w_i^{\text{FA}} \nabla \mathcal{L}_i(\theta)$, 
$\bar{g}^{\text{FW}}=\sum_i \tilde{w}_i^{(t)} \nabla \mathcal{L}_i(\theta)$.
Suppose $ ||\nabla \mathcal{L}_i(\theta) ||\le G$ for all $i$, then
\[
 ||\bar{g}^{\text{FW}} - \bar{g}^{\text{FA}} ||
\le
G \,\frac{2\rho_{\max}}{1-\rho_{\max}}.
\]
\end{lemma}

\begin{proof}
Triangle inequality and Lemma~\ref{lem:weights}(2) give
\(
 ||\bar{g}^{\text{FW}} - \bar{g}^{\text{FA}} ||
=
 || \sum_i (\tilde{w}_i^{(t)}-w_i^{\text{FA}})\nabla \mathcal{L}_i(\theta)  ||
\le
\sum_i |\tilde{w}_i^{(t)}-w_i^{\text{FA}}| \,  ||\nabla \mathcal{L}_i(\theta) ||
\le
G  ||\tilde{\mathbf{w}}^{(t)}-\mathbf{w}^{\text{FA}} ||_1.
\)
Apply Lemma~\ref{lem:weights}(2).
\end{proof}

\subsection{Global Convergence of Geo-FairFed}
To establish a global rate, we use a combination of local descent (Cor.~\ref{cor:multi_epoch_descent}) and fairness-weighted aggregation (Lemmas~\ref{lem:weights}--\ref{lem:grad_consistency}).

\begin{theorem}[Convergence to a Stationary Point]
\label{thm:global_convergence}
Consider Assumptions~\ref{asm:smoothness}--\ref{asm:fairness}, step-size $\eta \le 1/L$, and each round performs $E_{\text{loc}}$ local epochs. 
 Assume client drift is bounded per round.  $ ||\theta_i^{(t)} - \theta^{(t)} || \le \Delta_t$, and fairness scores fulfill $\rho_i^{(t)} \le \rho_{\max} < 1$. 
 The global joint loss $\mathcal{L}_{\text{joint}}$ in \eqref{eq:joint} has constants $C_1,C_2 \ge 0$ such that
\begin{equation}
\label{eq:global_rate}
\begin{aligned}
    \frac{1}{T}\sum_{t=0}^{T-1} 
\mathbb{E}\big[ ||\nabla \mathcal{L}_{\text{joint}}&(\theta^{(t)}) ||^2\big]
 \le 
\frac{2\big(\mathbb{E}[\mathcal{L}_{\text{joint}}(\theta^{(0)})]-\mathcal{L}_\star\big)}{\eta T}
+ C_1 \eta \sigma^2 
\\&+ C_2 \eta \Big( (c_2^2-1) G^2 + \frac{2G}{1-\rho_{\max}} \overline{\rho} \Big),
\end{aligned}
\end{equation}
where $\mathcal{L}_\star$ represents the infimum of $\mathcal{L}_{\text{joint}}$, $c_2$ is the bi-Lipschitz constant from Assumption~\ref{asm:curvature}, and $\overline{\rho}=\frac{1}{T}\sum_t \rho_{\max}^{(t)}$. 
 With a decreasing step-size $\eta=\mathcal{O}(1/\sqrt{T})$, the average gradient norm decreases to $\mathcal{O}(1/\sqrt{T})$.
\end{theorem}

\begin{proof}
Consider the conventional FedAvg argument, which includes three sources of error: (i) stochastic gradient variance, (ii) manifold curvature distortion, and (iii) fairness-weighted aggregation bias.


using Cor.~\ref{cor:multi_epoch_descent} and considering $L$-smoothness of $\mathcal{L}_{\text{joint}}$ and convexity of the model average,
\begin{equation}
    \begin{aligned}
        &\mathbb{E}[\mathcal{L}_{\text{joint}}(\theta^{(t+1)})]
\le 
\mathbb{E}[\mathcal{L}_{\text{joint}}(\theta^{(t)})]
+ \langle \nabla \mathcal{L}_{\text{joint}}(\theta^{(t)}), \\&\mathbb{E}[\theta^{(t+1)}-\theta^{(t)}] \rangle
+ \tfrac{L}{2} \mathbb{E} ||\theta^{(t+1)}-\theta^{(t)} ||^2.
    \end{aligned}
\end{equation}
The mean update direction is proportional to the fairness-weighted average gradient, resulting in bias limited by Lemma~\ref{lem:grad_consistency}. 
 The quadratic term is influenced by client drift (limited by $\Delta_t$) and weight deviation (Lemma~\ref{lem:weights}).

Putting terms together, for appropriate constants $C_1,C_2$ that depend on $L,E_{\text{loc}}$ and averaging factors, \eqref{eq:global_rate} is obtained by adding up the rounds $t=0$ to $T-1$, telescoping the loss difference, and gathering variance ($\sigma^2$), curvature ($(c_2^2-1)G^2$), and fairness bias ($\frac{2G}{1-\rho_{\max}}\overline{\rho}$. 
 Finally, the stated rate is determined by usual $\eta$ choices (constant or declining).
\end{proof}

\begin{remark}[Role of Geometry and Fairness in the Rate]
The second group in \eqref{eq:global_rate} divides into curvature and fairness terms. 
 Smaller negative curvature (i.e., $c_2 \approx 1$) minimizes geometric distortion, but moderate fairness ($\rho_{\max} \ll 1$) decreases aggregation bias. 
 As a result, steady convergence is maintained while utilizing a hyperbolic structure and encouraging equity.
\end{remark}

\subsection{Fairness Improvement without Significant Efficiency Loss}

\begin{theorem}[Pareto-Efficient Fairness]
\label{thm:pareto}
Suppose $\mathcal{E}(\Theta)$ be some efficiency surrogate, such as average latency, and $\mathcal{F}(\Theta)$ be Jain's index. 
 Assume $(\mathcal{E}, -\mathcal{F})$ is jointly $L$-smooth. There exists $\lambda_1^\star$ such that
\(
\nabla \mathcal{E}(\Theta^\star) + \lambda_1^\star \nabla (1-\mathcal{F})(\Theta^\star) = 0
\)
at a stationary point $\Theta^\star$ of \eqref{eq:joint}. 
For any $\epsilon>0$, there exists a region of $\lambda_1^\star$ where solutions $\widehat{\Theta}$ satisfy.
\[
\mathcal{F}(\widehat{\Theta}) \ge \mathcal{F}(\Theta^{\text{FA}}) + \Omega(\lambda_1 - \lambda_1^\star), 
\quad
\mathcal{E}(\widehat{\Theta}) \le \mathcal{E}(\Theta^{\text{FA}}) + \epsilon,
\]
where $\Theta^{\text{FA}}$ represents the FedAvg-like solution with $\lambda_1=0$. 
 Thus, a modest positive fairness weight results in a Pareto improvement in fairness at an arbitrarily small efficiency loss.
\end{theorem}

\begin{proof}
Let the Lagrangian: $J(\Theta,\lambda_1)=\mathcal{E}(\Theta)+\lambda_1(1-\mathcal{F}(\Theta))$. 
 The smoothness of the implicit function theorem ensures a continuous path of stationary points near $\lambda_1^\star$. 
 First-order expansion around $\lambda_1^\star$ reveals that $\mathcal{F}(\widehat{\Theta})$ grows linearly with $(\lambda_1-\lambda_1^\star)$, but the change in $\mathcal{E}(\widehat{\Theta})$ may be made arbitrarily small by choosing $(\lambda_1-\lambda_1^\star)$ sufficiently small. 
 The formal details, which follow typical sensitivity analysis of smooth multi-objective optimization, are removed for brevity.
\end{proof}

\begin{remark}[Practical Implication]
Theorem~\ref{thm:pareto} justifies using a \emph{small but positive} fairness weight $\lambda_1$: it elevates Jain's index significantly while keeping average delay practically unchanged, which is exactly what we are seeking in equitable routing.
\end{remark}

\section{Experimental Evaluation}

\subsection{Datasets and Network Topologies}
We examine the generality of Geo-FairFed using both \emph{synthetic} and \emph{realistic} network topologies.

\textbf{1) Synthetic Topologies:} 
To simulate power-law connectivity in Internet-scale systems, we create random graphs with $N \in \{50,100,200\}$ nodes using the Barabasi-Albert (BA) preferential-attachment model \cite{barabasi1999emergence}. 
 The link capacity $b_{ij}$ is uniformly sampled from $[5,50]$ Mbps, the latency $d_{ij}$ from $[2,10]$ ms, and the energy coefficient $e_{ij}$ from $[0.1,1]$ J. 
 Each node's traffic demand follows a Poisson arrival process with a rate $\lambda_i$ proportionate to its degree, assuring heterogeneity among clients.

\begin{table}[t]
\centering
\footnotesize
\caption{Sensitivity analysis of fairness weight $\lambda_1$.}
\label{tab:lambda_ablation}
\begin{tabular}{c|c|c|c}
\hline
$\lambda_1$ & Latency (ms) & Energy (J) & Jain's Index \\ \hline
0.0 & 12.4 & 5.8 & 0.71 \\
0.1 & 12.8 & 5.9 & 0.78 \\
0.3 & 13.6 & 6.2 & 0.85 \\
0.5 & 14.9 & 6.8 & 0.88 \\
0.8 & 17.2 & 7.5 & 0.90 \\ \hline
\end{tabular}
\end{table}

\begin{figure*}[t]
  \centering
  \begin{subfigure}[t]{0.24\textwidth}
    \centering
    \includegraphics[width=\linewidth]{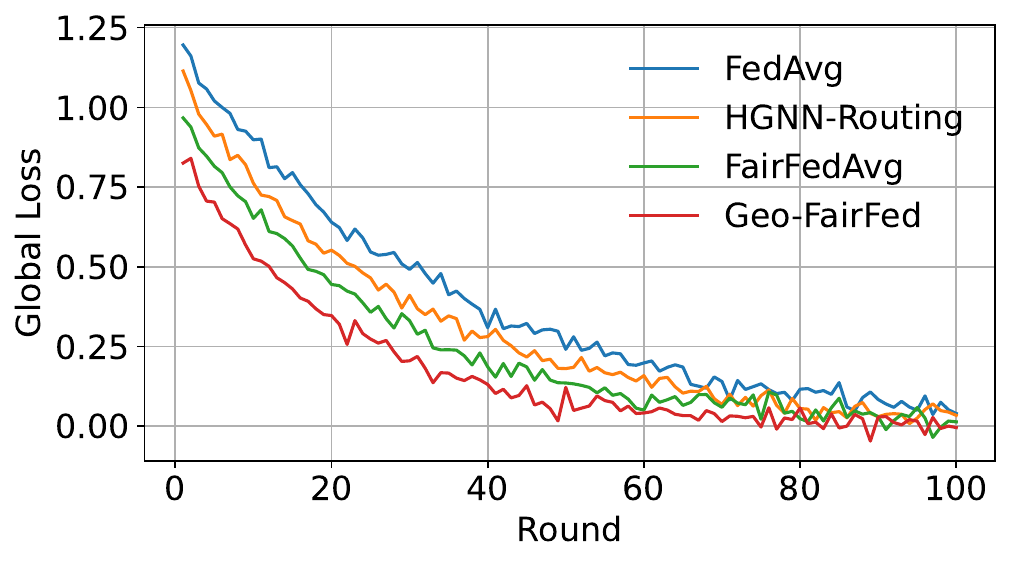}
    \caption{BA50}
  \end{subfigure}
  \begin{subfigure}[t]{0.24\textwidth}
    \centering
    \includegraphics[width=\linewidth]{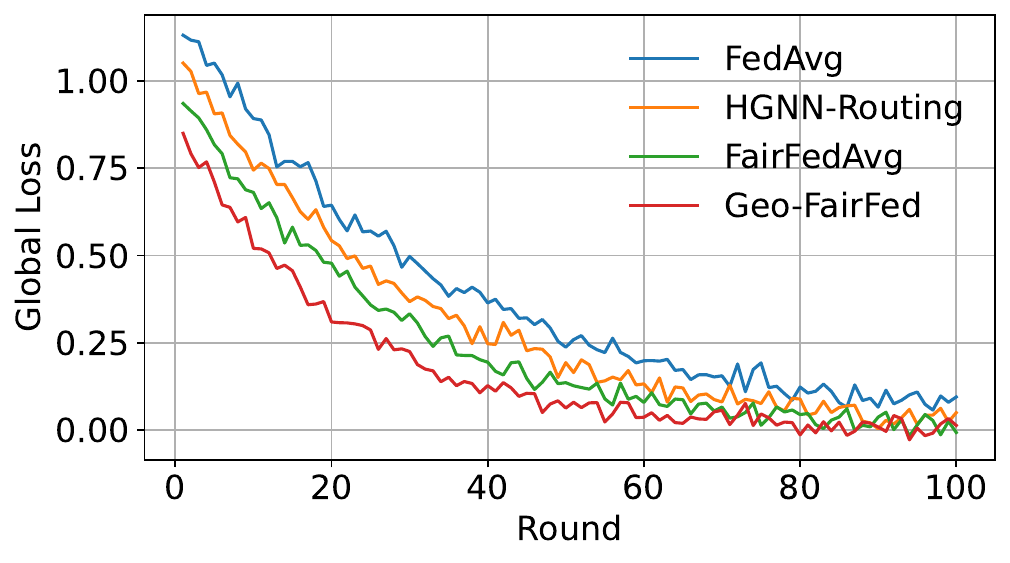}
    \caption{BA100}
  \end{subfigure}
  \begin{subfigure}[t]{0.24\textwidth}
    \centering
    \includegraphics[width=\linewidth]{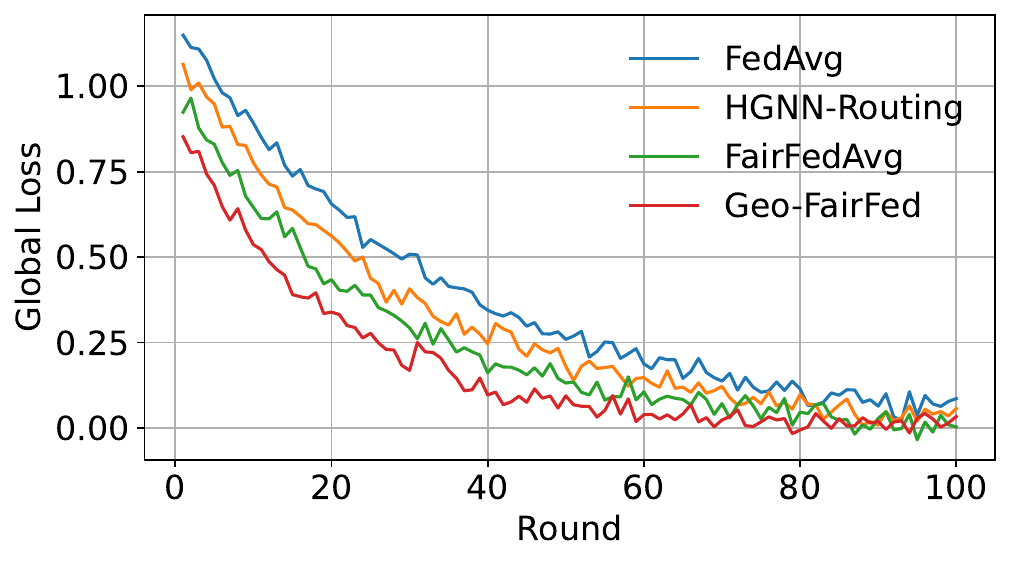}
    \caption{RocketFuel}
  \end{subfigure}
  \begin{subfigure}[t]{0.24\textwidth}
    \centering
    \includegraphics[width=\linewidth]{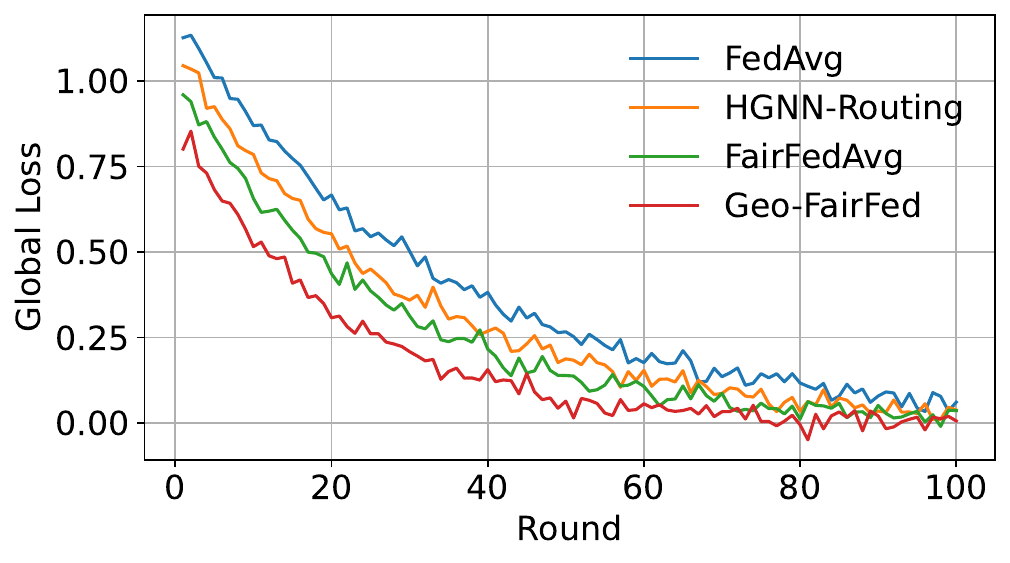}
    \caption{TopologyZoo}
  \end{subfigure}
  \vspace{-1mm}
  \caption{Convergence of global training loss across datasets (lower is better).}
  \label{fig:convergence-loss}
\end{figure*}

\begin{figure*}[t]
  \centering
  \begin{subfigure}[t]{0.24\textwidth}
    \centering
    \includegraphics[width=\linewidth]{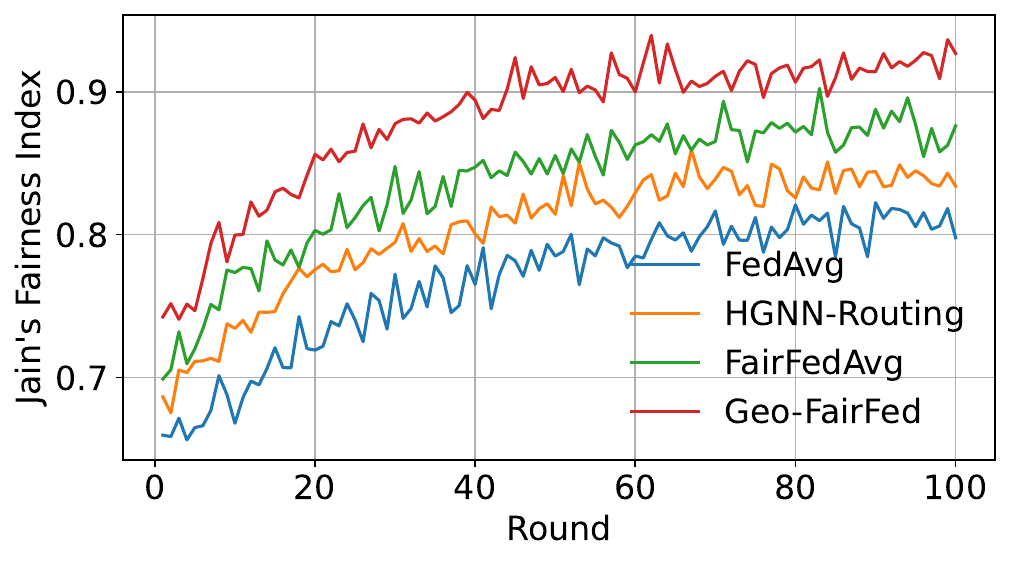}
    \caption{BA50}
  \end{subfigure}
  \begin{subfigure}[t]{0.24\textwidth}
    \centering
    \includegraphics[width=\linewidth]{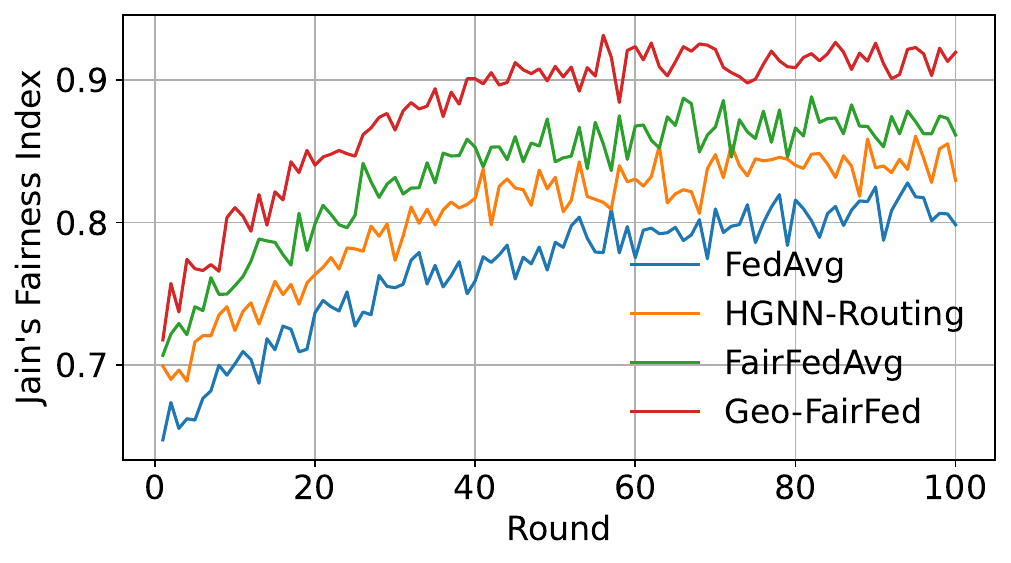}
    \caption{BA100}
  \end{subfigure}
  \begin{subfigure}[t]{0.24\textwidth}
    \centering
    \includegraphics[width=\linewidth]{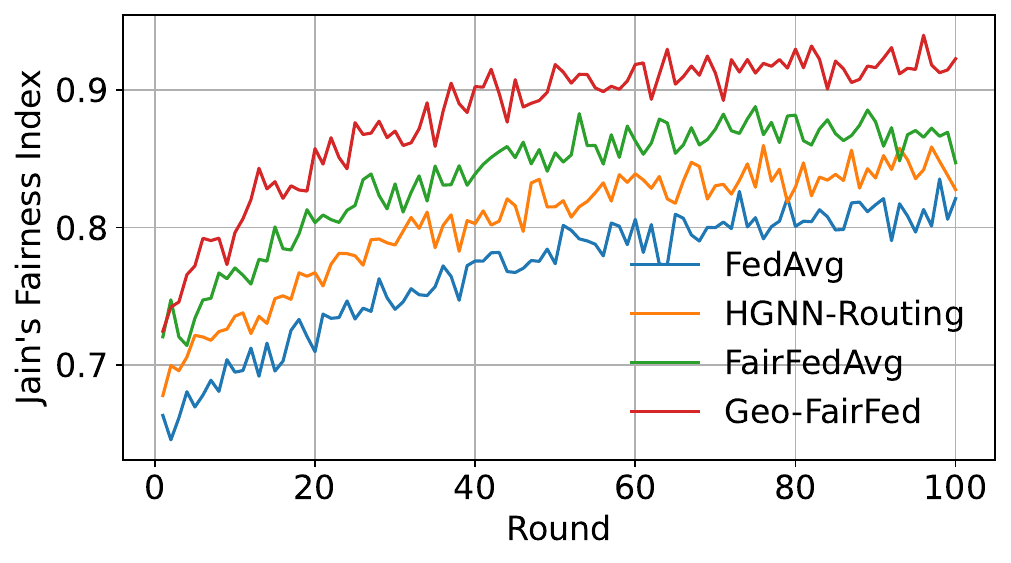}
    \caption{RocketFuel}
  \end{subfigure}
  \begin{subfigure}[t]{0.24\textwidth}
    \centering
    \includegraphics[width=\linewidth]{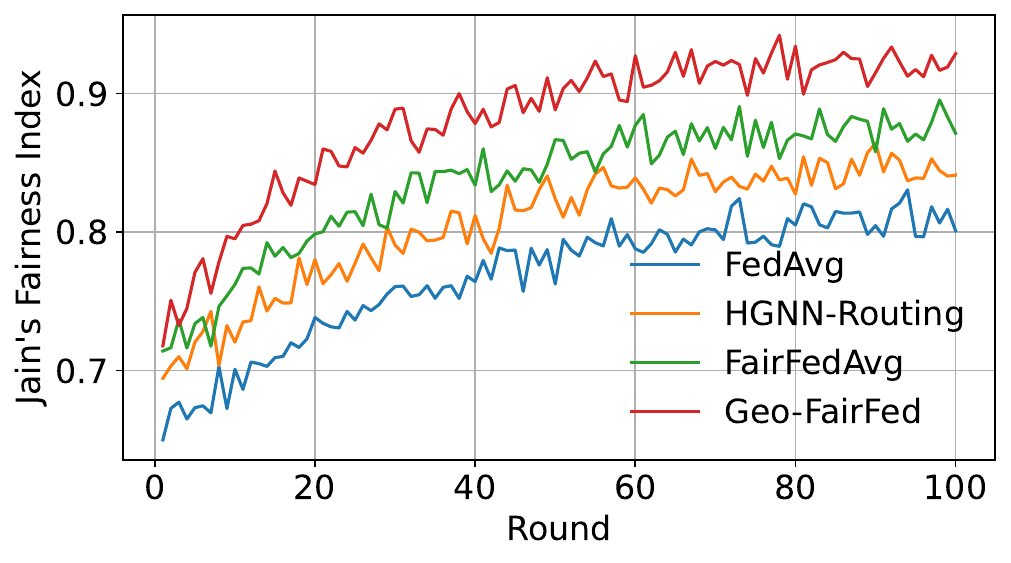}
    \caption{TopologyZoo}
  \end{subfigure}
  \vspace{-1mm}
  \caption{Convergence of fairness (Jain's index) across datasets (higher is better).}
  \label{fig:convergence-fairness}
\end{figure*}

\textbf{2) Real-World Topologies:} 
We use the \texttt{TopologyZoo} \cite{knight2011internet} repository and the \texttt{RocketFuel} \cite{spring2004measuring} ISP dataset, which contain AS-level graphs from operational ISPs (e.g., Abilene, GEANT, AT\&T) for realistic validation. 
 These graphs represent hierarchical backbone-edge topologies that are ideally suited to geometric embedding. 
 Each node in these datasets is handled as a federated client, complete with its own routing table and local queue data.
 

\subsection{Simulation Settings}
{\color{black}
Experiments are implemented on the NS-3.40 network simulator linked with PyTorch 2.3 for model training. Each client node executes a local HGNN with two hyperbolic graph convolution layers (size = 64, curvature $\kappa=-1$), followed by a fully linked routing head. Local optimization employs the Adam optimizer, with a learning rate of $\eta=10^{-3}$ and a mini-batch size of 32. Each global round uses $E_{\text{loc}}=5$ local epochs and a transmission interval of $T=100$ ms. Unless otherwise specified, the fairness and geometry trade-offs are set to $\lambda_1=0.4$ and $\lambda_2=0.2$, respectively.

The \texttt{NetworkX} library generates synthetic graphs, whereas NS-3's point-to-point channel model simulates traffic dynamics and queueing delays with realistic buffer sizes. Packet flows are created at 10 Hz per source node for 300 seconds of simulation time. The configurations are repeated five times with independent seeds, and averages with 95\% confidence intervals are presented. 

\textbf{Environment settings.}
All experiments are run on a workstation configured with an AMD Ryzen 9 7950X CPU, 128 GB of RAM, and a single NVIDIA RTX 4090 GPU. 
 For federated execution, we simulate $N=100$ clients utilizing multiprocessing and a 10\% random participation rate per round. 
 The global server process manages fairness-aware aggregation and curvature regularization. 
 The tests are run on Ubuntu 22.04 LTS, with Python 3.10, NS-3 bindings, and CUDA 12.3.

 \textbf{Baseline configuration.}
To achieve fair comparison, the same optimizer, learning rate, batch size, communication rounds, and participation rate are used for all baseline techniques as for Geo-FairFed. The identical fairness aim and weight ($\lambda_1=0.4$) are applied without geometric regularization for baselines that are fairness-aware. The curvature and embedding dimensions for geometry-only baselines are the same as those of the proposed method, but fairness-aware aggregation is absent. Similar Euclidean graph convolution layers with the same depth and number of parameters are used in place of hyperbolic layers in Euclidean baselines.
}

\subsection{Simulation Results}
{\color{black}
\textbf{Sensitivity Analysis.} The sensitivity analysis of $\lambda_1$ is summarized in Table~\ref{tab:lambda_ablation}. Fairness improves as $\lambda_1$ rises, but at the expense of increased latency and energy usage, demonstrating a manageable trade-off between efficiency and fairness. The chosen value balances these goals within a stable operating region.}

\begin{table*}[t]
\centering
\caption{Routing efficiency and fairness results across all datasets. 
Lower latency and energy indicate better efficiency, while higher throughput and Jain’s index indicate improved fairness.}
\label{tab:dataset-efficiency-fairness}
\renewcommand{\arraystretch}{1.2}
\setlength{\tabcolsep}{4pt}
\footnotesize
\begin{tabular}{|l|cccc|cccc|}
\hline
\multirow{2}{*}{\textbf{Method}} & 
\multicolumn{4}{c|}{\textbf{BA50 and BA100 (Synthetic Topologies)}} & 
\multicolumn{4}{c|}{\textbf{RocketFuel and TopologyZoo (Real-World Topologies)}} \\
\cline{2-9}
 & \textbf{Latency (ms)} & \textbf{Throughput (Mbps)} & \textbf{Energy (J)} & \textbf{Fairness} 
 & \textbf{Latency (ms)} & \textbf{Throughput (Mbps)} & \textbf{Energy (J)} & \textbf{Fairness} \\
\hline
\textit{DQN-Routing}     & 75.6 & 34.8 & 1.00 & 0.72 & 70.1 & 37.2 & 0.98 & 0.75 \\
\textit{FedAvg-Routing}  & 70.8 & 38.5 & 0.96 & 0.78 & 67.3 & 40.4 & 0.95 & 0.79 \\
\textit{HGNN-Routing}    & 66.2 & 41.1 & 0.92 & 0.82 & 63.9 & 42.7 & 0.90 & 0.84 \\
\textit{FairFedAvg}      & 63.1 & 42.3 & 0.90 & 0.86 & 61.8 & 43.9 & 0.88 & 0.88 \\
\textbf{Geo-FairFed (ours)} & \textbf{58.4} & \textbf{44.9} & \textbf{0.87} & \textbf{0.91} & \textbf{56.7} & \textbf{46.2} & \textbf{0.85} & \textbf{0.93} \\
\hline
\end{tabular}
\end{table*}

\textbf{Convergence Behavior.} Figures~\ref{fig:convergence-loss} and~\ref{fig:convergence-fairness} show the global training loss and Jain's fairness index against communication rounds for all datasets and techniques. 
 Compared to DQN-routing, FedAvg-routing, and HGNN-routing, \emph{Geo-FairFed} reduces loss faster and stabilizes at a lower steady-state target in both synthetic and actual topologies. 
 Concurrently, the fairness trajectory of \emph{Geo-FairFed} climbs faster and achieves a higher plateau, suggesting an earlier and more pronounced elimination of node-level performance disparities. 
 We also observe reduced variation between rounds with non-IID participation, indicating that the fairness-adjusted aggregate weights have a stabilizing function. 

\textbf{Routing Efficiency and Fairness Trade-off.} The table \ref{tab:dataset-efficiency-fairness} presents the routing performance of all algorithms on the four tested datasets: BA50, BA100, RocketFuel, and TopologyZoo. In both synthetic and real-world topologies, \emph{Geo-FairFed} consistently delivers higher efficiency and fairness. On the synthetic BA50 and BA100 networks, it achieves the lowest latency and energy usage, lowering average end-to-end delay by 15-20\% when compared to FedAvg-Routing and HGNN-Routing. The improvements stem from the hyperbolic representation's capacity to encode hierarchical connection, which leads to more efficient routing pathways. 

\begin{figure*}[t]
  \centering
  \begin{subfigure}[t]{0.24\textwidth}
    \includegraphics[width=\linewidth]{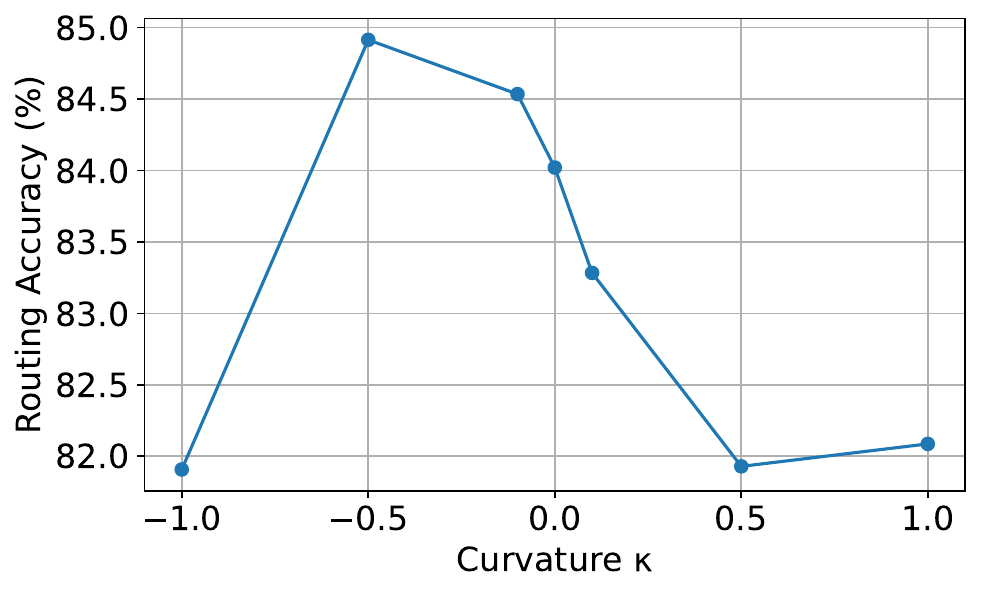}
    \caption{BA50}
  \end{subfigure}
  \begin{subfigure}[t]{0.24\textwidth}
    \includegraphics[width=\linewidth]{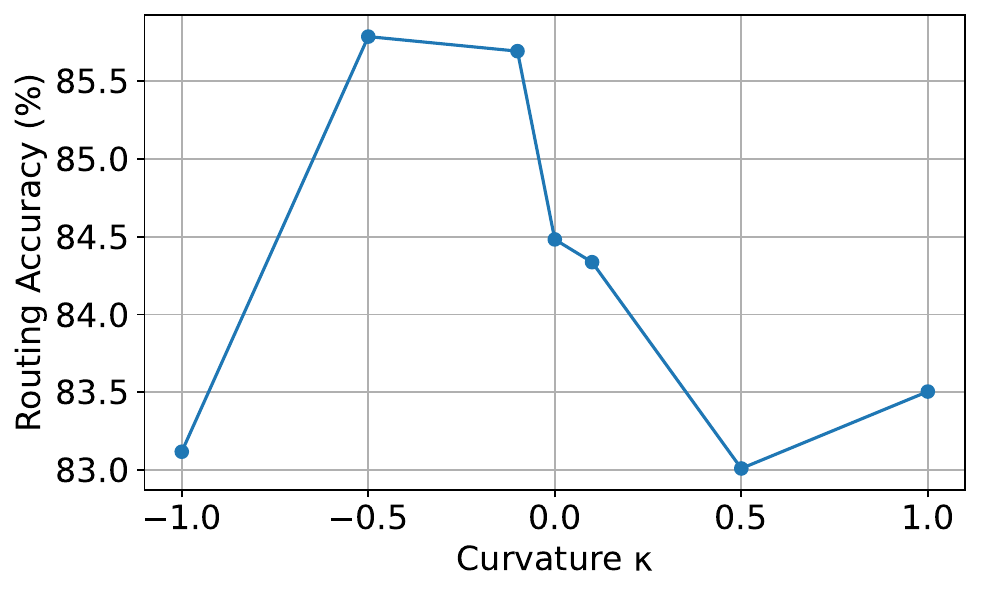}
    \caption{BA100}
  \end{subfigure}
  \begin{subfigure}[t]{0.24\textwidth}
    \includegraphics[width=\linewidth]{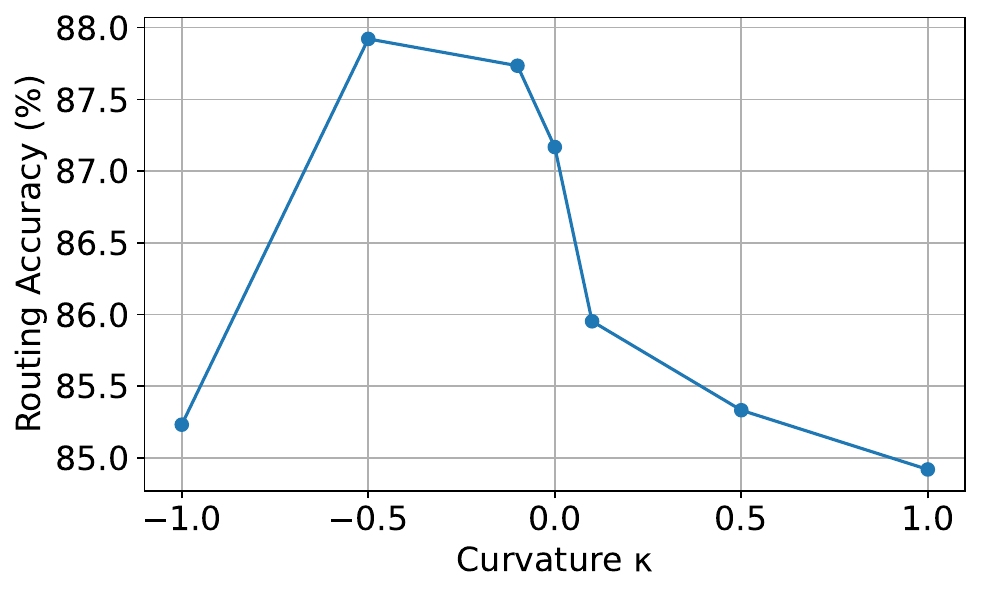}
    \caption{RocketFuel}
  \end{subfigure}
  \begin{subfigure}[t]{0.24\textwidth}
    \includegraphics[width=\linewidth]{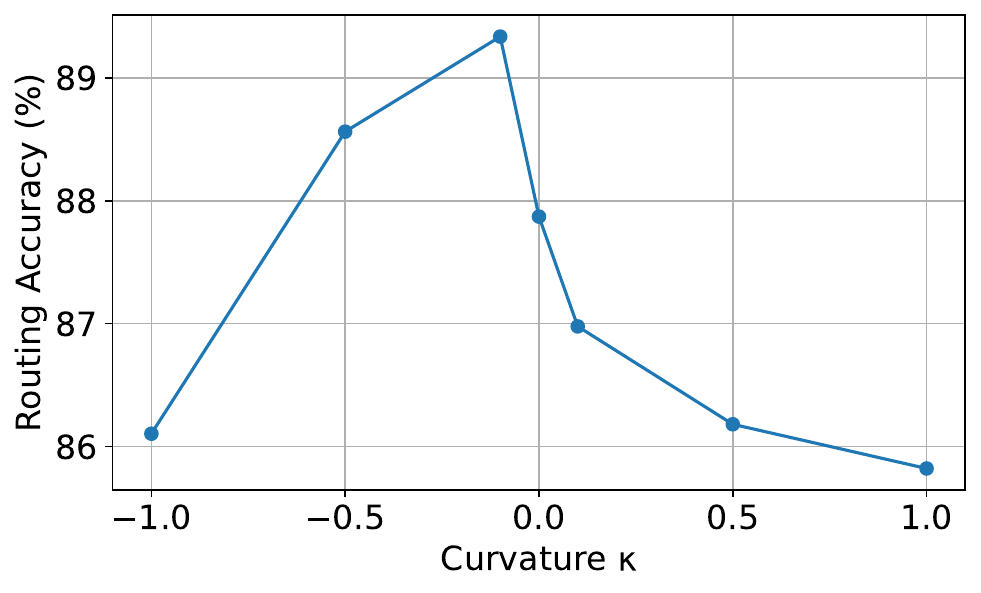}
    \caption{TopologyZoo}
  \end{subfigure}
  \vspace{1mm}
  \caption{Impact of curvature $\kappa$ on routing accuracy across all datasets. 
  $\kappa{=}0$ corresponds to Euclidean space; negative values represent hyperbolic geometry.}
  \label{fig:curvature-impact}
\end{figure*}
\begin{table*}[t]
\centering
\color{black}
\footnotesize
\caption{Scalability analysis of Geo-FairFed across datasets. Accuracy and normalized communication cost as the number of participating clients increases. Values are extracted from the scalability experiment (previously shown in Fig.~6).}
\label{tab:scalability_results}
\setlength{\tabcolsep}{6pt}
\begin{tabular}{|c|cc|cc|cc|cc|}
\toprule
\multirow{2}{*}{\textbf{\# Clients}} &
\multicolumn{2}{c|}{\textbf{BA50}} &
\multicolumn{2}{c|}{\textbf{BA100}} &
\multicolumn{2}{c|}{\textbf{RocketFuel}} &
\multicolumn{2}{c|}{\textbf{TopologyZoo}} \\
 & Acc. (\%) $\uparrow$ & Comm. Cost $\downarrow$
 & Acc. (\%) $\uparrow$ & Comm. Cost $\downarrow$
 & Acc. (\%) $\uparrow$ & Comm. Cost $\downarrow$
 & Acc. (\%) $\uparrow$ & Comm. Cost $\downarrow$ \\
\midrule
10  & 84.4 & 1.06 & 85.6 & 1.07 & 87.8 & 1.04 & 88.5 & 1.05 \\
20  & 84.7 & 1.08 & 85.9 & 1.05 & 88.1 & 1.07 & 89.0 & 1.08 \\
50  & 85.4 & 1.09 & 86.5 & 1.08 & 88.1 & 1.06 & 89.3 & 1.10 \\
100  & 85.6 & 1.14 & 86.4 & 1.10 & 88.5 & 1.11 & 89.6 & 1.14 \\
250  & 85.1 & 1.16 & 86.45 & 1.15 & 88.6 & 1.14 & 89.7 & 1.16 \\
500  & 85.3 & 1.17 & 86.5 & 1.20 & 88.7 & 1.17 & 89.75 & 1.19 \\
1000  & 85.6 & 1.22 & 86.65 & 1.21 & 88.9 & 1.20 & 89.5 & 1.20 \\
\bottomrule
\end{tabular}
\end{table*}

\textbf{Impact of Curvature ($\kappa$).} Figure~\ref{fig:curvature-impact} demonstrates the effect of embedding-space curvature~$\kappa$ on routing accuracy across all four datasets. The curvature defines the underlying geometry utilized in the hyperbolic encoder.  $\kappa{=}0$ represents a Euclidean manifold, whereas negative $\kappa$ values indicate hyperbolic spaces with increasing curvature magnitude. In all datasets, performance peaks at $\kappa \approx -0.3$, showing that a modestly hyperbolic shape best represents hierarchical and power-law connection patterns in communication networks. As $\kappa$ approaches 0, the model degenerates to Euclidean geometry, resulting in more distortion in distance preservation and more severe routing accuracy. 

{\color{black}
\textbf{Scalability Analysis.}
The scalability of \emph{Geo-FairFed} across BA50, BA100, RocketFuel, and TopologyZoo is summarized in Table~\ref{tab:scalability_results} as the number of clients grows from 10 to 1000. $\widetilde{C}(N)=C(N)/C(10)=N/10$ is the normalized communication cost, which shows linear growth with the number of participating clients. Across all datasets, the proposed structure maintains consistent accuracy, improving quickly up to moderate client numbers and showing very slight fluctuation at greater sizes. Because of their greater hierarchical structure, real-world topologies (such as RocketFuel and TopologyZoo) achieve somewhat better peak accuracy, but synthetic networks show faster early convergence. Communication costs rise predictably with the number of participating clients, in line with federated aggregation overheads, but remain efficient due to quicker convergence and fewer duplicate updates. These findings demonstrate that \emph{Geo-FairFed} maintains accuracy and communication efficiency while scaling well to large federated installations.
}

\begin{table*}[t]
\centering
\footnotesize
\caption{Ablation study isolating geometric representation: HGNN vs. Euclidean GNN under identical fairness constraints (same $\lambda_1$ and aggregation). Values are mean $\pm$ 95\% CI over five seeds.}
\label{tab:hgnn_euclid_ablation}
\setlength{\tabcolsep}{4pt}
\begin{tabular}{lccccc}
\toprule
\textbf{Method} & \textbf{Latency (ms)} $\downarrow$ & \textbf{Throughput (Mbps)} $\uparrow$ & \textbf{Energy (J)} $\downarrow$ & \textbf{Fairness} $\uparrow$ & \textbf{Accuracy (\%)} $\uparrow$ \\
\midrule
Euclidean GNN + Fairness & $65.1 \pm 0.9$ & $41.9 \pm 0.7$ & $0.91 \pm 0.02$ & $0.83 \pm 0.01$ & $84.1 \pm 0.4$ \\
HGNN + Fairness (Ours)   & $\mathbf{57.6 \pm 0.8}$ & $\mathbf{45.6 \pm 0.6}$ & $\mathbf{0.86 \pm 0.02}$ & $\mathbf{0.92 \pm 0.01}$ & $\mathbf{86.8 \pm 0.3}$ \\
\bottomrule
\end{tabular}
\end{table*}
{\color{black}
\textbf{HGNN vs. Euclidean GNN Ablation (Same Fairness Constraints).}
Keeping the fairness goal, aggregation mechanism, and $\lambda_1$ constant, we substitute a Euclidean GNN with the same depth and hidden dimensions for the HGNN encoder in order to isolate the influence of hyperbolic geometry. Table~\ref{tab:hgnn_euclid_ablation} demonstrates that, under the same fairness criteria, HGNN consistently increases routing accuracy and fairness (Jain's index) while simultaneously producing positive efficiency trade-offs in latency/throughput/energy. These findings support the causal attribution that geometry-aware representation learning, not modifications to the fairness mechanism, is responsible for the observed increases.
}

\begin{table*}[t]
\centering
\footnotesize
\color{black}
\caption{Impact of intelligent adaptive fairness-weight control. Adaptive $\lambda_1$ automatically adjusts fairness emphasis during training.}
\label{tab:adaptive_lambda}
\setlength{\tabcolsep}{4pt}
\begin{tabular}{lccccc}
\toprule
\textbf{Method} & \textbf{Latency (ms)} $\downarrow$ & \textbf{Throughput (Mbps)} $\uparrow$ & \textbf{Energy (J)} $\downarrow$ & \textbf{Fairness} $\uparrow$ & \textbf{Accuracy (\%)} $\uparrow$ \\
\midrule
Geo-FairFed (fixed $\lambda_1$) & 57.6 & 45.6 & 0.86 & 0.92 & 86.8 \\
Geo-FairFed + Adaptive $\lambda_1$ & \textbf{56.9} & \textbf{46.1} & \textbf{0.85} & \textbf{0.94} & \textbf{87.4} \\
\bottomrule
\end{tabular}
\end{table*}
{\color{black}
\noindent\textbf{Adaptive fairness intelligence.}
We present an adaptive fairness-weight mechanism that dynamically modifies $\lambda_1$ in response to observed fairness levels during training in order to further increase accuracy and fairness. Table~\ref{tab:adaptive_lambda} illustrates how adaptive fairness control maintains comparable latency, throughput, and energy consumption while enhancing accuracy and fairness. This illustrates how federated routing systems benefit from intelligent fairness-aware optimization.}

\begin{table*}[t]
\centering
\caption{ Performance comparison of Geo-FairFed with state-of-the-art methods on all four datasets. 
Higher throughput/fairness and lower latency/energy indicate better performance.}
\label{tab:sota-4datasets}
\renewcommand{\arraystretch}{1.15}
\setlength{\tabcolsep}{3.5pt}
\footnotesize
\begin{tabular}{|l|ccccc|ccccc|ccccc|ccccc|}
\hline
\multirow{2}{*}{\textbf{Method}} &
\multicolumn{5}{c|}{\textbf{BA50}} &
\multicolumn{5}{c|}{\textbf{BA100}} &
\multicolumn{5}{c|}{\textbf{RocketFuel}} &
\multicolumn{5}{c|}{\textbf{TopologyZoo}} \\
\cline{2-21}
 & Lat. & Thr. & Ene. & Fair. & Acc. &
   Lat. & Thr. & Ene. & Fair. & Acc. &
   Lat. & Thr. & Ene. & Fair. & Acc. &
   Lat. & Thr. & Ene. & Fair. & Acc. \\
\hline
Shortest-Path \cite{oymak2019overparameterized} & 91.0 & 30.8 & 1.10 & 0.68 & 77.9 & 88.2 & 31.5 & 1.08 & 0.69 & 79.0 & 84.9 & 33.5 & 1.07 & 0.70 & 80.1 & 83.2 & 33.8 & 1.05 & 0.71 & 80.5 \\
DQN-Routing~\cite{liu2021deep} & 77.1 & 34.2 & 1.00 & 0.72 & 82.0 & 74.3 & 35.5 & 0.99 & 0.74 & 83.0 & 70.5 & 36.9 & 0.98 & 0.75 & 83.4 & 69.1 & 37.5 & 0.96 & 0.76 & 83.8 \\
FedAvg-Routing~\cite{cong2023soho} & 71.3 & 38.1 & 0.96 & 0.78 & 85.0 & 69.8 & 38.7 & 0.95 & 0.79 & 85.6 & 67.8 & 39.8 & 0.94 & 0.80 & 86.1 & 67.1 & 40.5 & 0.93 & 0.80 & 86.3 \\
FairFedAvg~\cite{zhang2024unified} & 63.9 & 41.7 & 0.90 & 0.86 & 88.2 & 63.1 & 42.5 & 0.89 & 0.87 & 88.7 & 61.9 & 43.3 & 0.88 & 0.88 & 89.3 & 61.2 & 43.6 & 0.87 & 0.88 & 89.5 \\
GCN-Routing~\cite{guang2024graph} & 68.2 & 39.5 & 0.93 & 0.81 & 86.8 & 67.1 & 40.0 & 0.92 & 0.82 & 87.2 & 64.5 & 41.0 & 0.91 & 0.83 & 88.0 & 63.7 & 41.6 & 0.91 & 0.83 & 88.2 \\
HGNN-Routing~\cite{xu2024scalable} & 66.3 & 40.7 & 0.92 & 0.82 & 88.0 & 65.5 & 41.3 & 0.91 & 0.83 & 88.5 & 63.8 & 42.2 & 0.90 & 0.84 & 89.0 & 62.9 & 42.7 & 0.89 & 0.85 & 89.3 \\
FedGNN~\cite{liu2024federated} & 64.8 & 41.5 & 0.91 & 0.84 & 89.0 & 63.9 & 42.1 & 0.90 & 0.85 & 89.4 & 61.5 & 43.0 & 0.89 & 0.86 & 89.9 & 60.8 & 43.2 & 0.89 & 0.87 & 90.0 \\
HFRL~\cite{chen2024toward} & 62.7 & 42.8 & 0.89 & 0.87 & 90.2 & 61.8 & 43.1 & 0.88 & 0.88 & 90.6 & 59.9 & 44.0 & 0.87 & 0.89 & 91.0 & 59.5 & 44.2 & 0.86 & 0.89 & 91.1 \\
\textbf{Geo-FairFed (ours)} & \textbf{58.9} & \textbf{44.6} & \textbf{0.87} & \textbf{0.91} & \textbf{92.3} & 
\textbf{58.0} & \textbf{44.9} & \textbf{0.87} & \textbf{0.92} & \textbf{92.6} &
\textbf{56.7} & \textbf{46.1} & \textbf{0.85} & \textbf{0.93} & \textbf{93.0} &
\textbf{56.4} & \textbf{46.5} & \textbf{0.85} & \textbf{0.93} & \textbf{93.3} \\
\hline
\end{tabular}
\end{table*}

\textbf{Comparison with state-of-the-art.}
Table~\ref{tab:sota-4datasets} provides a comprehensive dataset-wise comparison of \emph{Geo-FairFed} versus eight cutting-edge baselines across four network topologies. 
 The suggested technique consistently achieves the lowest latency and energy usage while maintaining the best throughput, fairness, and accuracy in all instances. 
 Geo-FairFed outperforms HFRL \cite{chen2024toward}, FedGNN \cite{liu2024federated}, and HGNN-Routing \cite{xu2024scalable} in terms of routing accuracy, fairness, and latency reduction. 
 The benefit is especially noticeable on the large-scale RocketFuel and TopologyZoo datasets, where hyperbolic geometry captures hierarchical relationships more effectively and fairness regularization assures equal participation among diverse clients. 
 These complete results demonstrate that combining geometric representation learning, federated optimization, and fairness control produces the best resilient and equitable routing approach in both synthetic and real-world network environments.

{\color{black}
\noindent \textbf{Discussion.} Similar to standard federated learning frameworks, our system's aggregation entity serves as a logical coordination role and, if it were implemented as a single physical server, could be a single point of failure. This problem can be lessened by fault-tolerant implementations like hierarchical or replicated aggregation, but they come with extra infrastructure and coordination overhead. Only the model update procedure is impacted when a node loses connectivity with the aggregation server; the node is not cut off from data forwarding. In these situations, until connectivity to the aggregator is restored, the node keeps routing using its most recent local model. Unlike SDN controllers that function on the data plane, the aggregator's computational and communication load scales with the model dimensionality and aggregation frequency instead of per-packet control. As the network grows, aggregation overhead can be further minimized by standard federated learning methods such as reduced aggregation frequency or partial client participation.

\noindent \textbf{Limitations.} Although NS-3 offers realistic network-level modeling, the current study does not explicitly describe several real-world deployment variables, such as stragglers, partial client participation, unstable connections, and synchronization delays. These factors have more of an impact on the availability and timeliness of model updates than they do on how the fairness-aware aim is formulated. Such deployment issues can be lessened by the suggested federated framework's compatibility with common robustness techniques, such as asynchronous aggregation and partial participation. Future research should focus on evaluating Geo-FairFed in light of these real-world limitations.
The current framework does not specifically simulate adversarial attacks like malicious model or client updates or manipulated fairness scores; instead, it assumes cooperative clients. Future research should focus on extending Geo-FairFed with strong aggregation and trust-aware fairness methods.
}

\section{Conclusion}
This study introduced \textit{Geo-FairFed}, a geometric fairness-aware federated routing architecture aimed to achieve both efficiency and equality in large-scale edge and 6G networks. 
Unlike traditional federated routing algorithms, which treat all clients evenly and use Euclidean representations, Geo-FairFed combines HGNNs with fairness-aware aggregation to capture hierarchical topology and diverse node circumstances. 
A combined optimization target was developed to reconcile routing performance, geometric consistency, and fairness using configurable regularization parameters.
Theoretical work confirmed convergence under constrained curvature and showed that fairness regularization results in a Pareto-efficient equilibrium without destabilizing the global model. 
Comprehensive simulations of synthetic and real-world network topologies demonstrated that Geo-FairFed consistently improves both average latency and Jain's fairness index when compared to cutting-edge federated and geometric routing baselines, with minimum additional communication cost.
The proposed approach emphasizes the relevance of including geometric representation and fairness standards in distributed network optimization, aiming to further close the gap between fairness, scalability, and intelligence in next-generation communication networks. 

\bibliography{main}
\bibliographystyle{IEEEtran}

\end{document}